\documentclass{article}

\usepackage[a4paper,margin=1in]{geometry}
\usepackage[english]{babel}
\usepackage{amsmath, amssymb, amsthm, mathtools,mathrsfs}
\usepackage{xcolor}
\usepackage[colorlinks]{hyperref}
\usepackage{doi}
\definecolor{darkred}{rgb}{0.5,0,0}
\definecolor{darkblue}{rgb}{0,0,0.5}
\definecolor{darkgreen}{rgb}{0,0.5,0}
\hypersetup{colorlinks
  ,linkcolor=darkred
  ,filecolor=darkgreen 
  ,urlcolor=darkred
  ,citecolor=darkblue,
  linktocpage}
\usepackage{cleveref}
\usepackage{centernot}
\usepackage{microtype}
\usepackage[inline]{enumitem}
\usepackage{doi}
\usepackage{dsfont}
\usepackage[nocompress]{cite}

\setlist[itemize]{label=-}

\theoremstyle{definition}
\newtheorem{definition}{Definition}[section]
\newtheorem{example}[definition]{Example}
\newtheorem{remark}[definition]{Remark}

\theoremstyle{plain}
\newtheorem{lemma}[definition]{Lemma}
\newtheorem{corollary}[definition]{Corollary}
\newtheorem{theorem}[definition]{Theorem}

\DeclareMathOperator{\conc}{conc}
\newcommand{\Problim}{\mathds{P}_{\lim}}
\newcommand{\Problimsup}{\overline{\mathds{P}}_{\lim}}
\newcommand{\sthat}{\mathbin |}
\DeclareMathOperator{\infix}{infix}
\newcommand{\cA}{\mathcal{A}}
\newcommand{\cB}{\mathcal{B}}
\newcommand{\abs}[1]{{\vert #1 \vert}}

\newcommand{\qtext}[1]{\quad\text{#1}\quad}
\newcommand{\qqtext}[1]{\qquad\text{#1}\qquad}
\newcommand{\N}{\mathbb{N}}%
\newcommand{\Prob}{\mathbb{P}}%

\renewcommand{\phi}{\varphi}

\newcommand{\emptyword}{\varepsilon}

\begin{document}

\title{Deciding Sparseness of Regular Languages of Finite Trees and Infinite Words}

\author{Kord Eickmeyer and Georg Schindling \\ Technische Universität Darmstadt}

\maketitle

\begin{abstract}
  We study the notion of sparseness for regular languages over finite trees and infinite words. 
  A language of trees is called sparse if the relative number of $n$-node trees in the language tends to zero, and
  a language of infinite words is called sparse if it has measure zero in the Bernoulli probability space.
  We show that sparseness is decidable for regular tree languages and for regular languages of infinite words. 
  For trees, we provide characterisations in terms of forbidden subtrees and tree automata, leading to a linear time decision procedure. 
  For infinite words, we present a characterisation via infix completeness and give a novel proof of decidability. 
  Moreover, in the non-sparse case, our algorithm computes a measurable subset of accepted words that can serve as counterexamples in almost-sure model checking.
  Our findings have applications to automata based model checking in formal verifications and XML schemas, among others.
\end{abstract}

    \section{Introduction}

We call a language of finite objects \emph{sparse} if the relative
number of $n$-element objects in the language tends to zero: For a
language $L \subseteq \Sigma^*$ of words we set
\(
  \Prob_n(L) \coloneqq \frac{\abs{L \cap \Sigma^n}}{\abs{\Sigma^n}}
\)
and call $L$ sparse if $\lim_{n\to\infty}\Prob_n(L) = 0$. Likewise, if
$L_T$ is a language of $\Sigma$-labelled trees and $\cB_n^\Sigma$ is the
set of all such trees with $n$ nodes, we set
\(
  \Prob_n(L_T) \coloneqq \frac{\abs{L_T \cap \cB_n^\Sigma}}{\abs{\cB_n^\Sigma}}
\)
and call $L_T$ sparse if $\lim_{n \to \infty}\Prob_n(L_T) = 0$. For a language
$L_\omega \subseteq \Sigma^\omega$ of infinite words we turn $\Sigma^\omega$
into a Bernoulli probability space and call $L_\omega$ sparse if it has
measure $0$. 

Regular languages play an important rôle both in applied and in
theoretical computer science. Apart from appearing in practical
applications such as pattern matching, one reason for their importance
from a theoretical perspective is that they allow for several
seemingly unrelated characterisations, in particular by finite
monoids, finite automata, and as being definable in monadic
second-order logic (MSO) by the Büchi-Elgot-Trakhtenbrot theorem~\cite{Buechi1960,Elgot1961,Trakhtenbrot1961}. 
Many problems that are undecidable or at
least computationally hard for general classes of languages are
computable for regular languages. In particular the emptiness problem
and, more generally, the subset relation between regular languages is
efficiently decidable.

Naturally, the techniques used in dealing with regular languages have
been adapted to other scenarios. Two of the most important and
successful generalisations have been to languages of finite trees and
languages of infinite words, also called $\omega$-languages. In both
these settings, the correspondence between regular languages, automata,
and definability in MSO still holds (see~\cite{Doner1970, Thatcher1968} and~\cite{Buechi1990}), yielding similar algorithmic
properties as for regular languages of finite words.

We provide a characterisation of sparse regular tree languages via 
\emph{forbidden subtrees} and \emph{tree automata} (\Cref{thm:tree:condition}). 
As a consequence, we deduce that sparseness of regular tree
languages is decidable in linear time (\Cref{cor:trees:decidable}).
Furthermore, our results also provide a method for proving 
non-regularity of sparse tree languages (\Cref{cor:non-regular}).
Finally, we observe that our characterisation by forbidden subtrees does not hold
for non-uniform models of random binary trees, specifically the random binary 
search tree model (\Cref{ex:bst}).

A major motivation for regular $\omega$-languages, on the other hand,
is in model checking of reactive systems~\cite{Baier2008}. In this
setting, we are given a formal specification of a system (e.g.~some
piece of hardware or software) and a specification of properties that the
system should have or not have. If both the system and the property
are specified in linear temporal logic (LTL), model checking
essentially reduces to checking whether a certain regular
$\omega$-language is empty or not.
Theorem~\ref{thm:infwords} extends this to \emph{almost sure model
  checking}, i.e.~the question whether a certain system \emph{almost
  surely} satisfies a given property. If the system fails to do so, it
is even possible to compute a counterexample, in this case a set of
$\omega$-words of strictly positive measure that almost surely violates
the condition to be checked.

Note that because the complement of a regular language is again
regular, and $L$ is sparse if and only if its complement has
(asymptotic) probability $1$, the same algorithms can be used to check
whether a language of trees or $\omega$-words has (asymptotic)
probability $1$.

\subsection*{Techniques and Related Work}
Sparseness of regular languages was investigated by
Sin'ya~\cite{Sinya2015}. He showed that a regular language
$L \subseteq \Sigma^*$ over a finite alphabet $\Sigma$ is sparse if,
and only if, $L$ has an excluded factor, i.e.
\(
  \Sigma^* x \Sigma^* \cap L = \emptyset
\)
for some $x \in \Sigma^*$. The proof of this characterisation was 
simplified by Koga in~\cite{Koga2019}. 
We build on their insights and generalize them significantly to handle 
the richer settings of regular tree languages and $\omega$-languages.

The trees we are most interested in are \emph{rooted}, \emph{ordered},
\emph{unranked} trees, i.e.~trees with a distinguished root node
in which the number of children of a node is unbounded and for every
node there is a linear order on the set of its children. The terms
\emph{siblinged} or \emph{planar} are sometimes used for what we call
ordered, e.g. in~\cite{Flajolet2009,Benedikt09}. We state our results
on regular tree languages in terms of binary trees with distinguished
left and right children, cf.~Section~\ref{sub:treelangs}. This allows
for a simpler definition of tree automata and comes at no loss of
generality because ordered unranked trees are bijectively
MSO-interpretable in binary trees, a folklore technique that is
spelled out, for example, in~\cite[Sec.~10.2]{Flum2006}.

Because regular languages exactly capture the expressive power of
monadic second-order logic (MSO), the (asymptotic) probability of
regular languages can be rephrased as the (asymptotic) probability
that a random structure satisfies a certain MSO sentence. This has
been well studied: For every first-order logic (FO) sentence with
relational signature, the probability that a random labelled structure
satisfies the sentence converges to either $0$ or
$1$ (this was proved independently in~\cite{Glebskij1969}
and~\cite{Fagin1976}, cf.~\cite[Ch.~4]{ebbflu99}), a result that has
been extended by Spencer~\cite{Spencer2013} to Erdős-Rényi graphs with
more general edge probabilities. The fact that every sentence is
satisfied asymptotically with probability $0$ or $1$ is usually called
\emph{$0$-$1$ law}.
For monadic second-order logic (MSO), McColm~\cite{McColm2002} showed
a 0-1 law for MSO on labelled trees. Note that McColm (crucially)
considers trees without a distinguished root, while our trees are
rooted. More recently, Malyshkin and Zhukovskii~\cite{Malyshkin2021}
showed a 0-1 law for MSO on finite trees with probability measures
other than uniform probabilities, namely uniform and preferential
attachment. Further results on the limit behaviour of MSO sentences on
various graph classes can be found in~\cite{Dawar2019}
and~\cite{Atserias2018}.
These results are not directly applicable to our setting. In fact,
there can be no $0$-$1$ law for regular languages: Both in trees and
in $\omega$-words there is some designated element (such as the first
position in a word or the root of a tree), and the set of all
structures in which this element gets a specific label is a regular
set with (asymptotic) density strictly between $0$ and $1$.

Recently, Niwiński et al.~\cite{Niwinski2023} proved that the exact
probability of acceptance for MSO formulae in infinite random trees is
computable. However, this can not be used to decide sparseness for
languages of finite trees or infinite words.

A necessary and sufficient condition for an $\omega$-language to have
strictly positive measure has already been proved by
Staiger~\cite{Staiger1998}. However, Staiger does not address the
question of decidability of this property, and his proof uses deep
results from topology and Kolmogorov complexity, namely
a connection between the \emph{subword complexity} of an $\omega$-word 
$\xi$ and the \emph{Hausdorff dimension} of the regular $\omega$-languages
containing $\xi$. Courcoubetis et al.~\cite{Courcoubetis1995} treated the 
sparseness problem for $\omega$-languages under the name \emph{probabilistic emptyness} 
and showed it to be decidable with methods similar to ours. By relating
non-sparseness to infix completeness, our (shorter) proof has the
additional feature of giving an exact characterisation of non-sparse
$\omega$-languages. Furthermore, the algorithm we present in
Thm.~\ref{thm:infwords} is the first to provide a concise
representation of a set $M$ of $\omega$-words that is almost a subset
of $L$, in the sense that $\Prob(M \setminus L) = 0$. This is
extremely valuable in the context of LTL model checking, as it amounts
to a non-negligible set of counterexamples.
\section{Preliminaries}
\label{sec:preliminaries}
\subsection{Languages and Automata}
By $\Sigma$ we denote a finite alphabet assuming $\abs{\Sigma} > 1$ to
avoid trivialities. For $k \geq 0$ we denote by $\Sigma^\ast$,
$\Sigma^k$, and $\Sigma^{\leq k}$ the sets of all strings, those of
length $k$, and of length at most $k$, respectively. The empty string
is denoted by $\emptyword$.
A \emph{language} is any subset $L \subseteq \Sigma^*$. Concatenation of
strings $u,v \in \Sigma^*$ is denoted by $u\cdot v$ or just $uv$, and
for a language $L$ we set $L^* \coloneqq \{ u_1\cdots u_k \sthat k \geq 0,
u_1,\ldots,u_k \in L\}$.

An \emph{$\omega$-word} over the alphabet $\Sigma$ is an infinite
sequence $a_0a_1a_2\cdots$ of letters $a_i \in \Sigma$. The set of all
$\omega$-words is denoted by $\Sigma^\omega$, and an $\omega$-language
is any subset $L \subseteq \Sigma^\omega$. A (finite) word $u \in
\Sigma^*$ and an $\omega$-word $v \in \Sigma^\omega$ may be
concatenated to $uv \in \Sigma^\omega$, and for a language $U
\subseteq \Sigma^*$ and an $\omega$-language $V \subseteq
\Sigma^\omega$ we set
\[
  U^\omega \coloneqq \{ u_1u_2u_3\cdots \sthat u_1,u_2,\ldots \in U \}
  \qqtext{and}
  UV \coloneqq \{ uv \sthat u \in U \text{ and }v \in V \}.
\]
A word $u \in \Sigma^*$ is a \emph{prefix} of $v \in \Sigma^* \cup
\Sigma^\omega$, written $u \preceq v$, if $v = uw$ for some $w \in
\Sigma^* \cup \Sigma^\omega$.

A \emph{deterministic finite automaton (DFA)} over the alphabet
$\Sigma$ is a tuple $\cA = (\Sigma,Q,q_0,\delta,A)$ consisting of a finite
set $Q$ of \emph{states}, a designated \emph{initial state} $q_0 \in
Q$, a \emph{transition function} $\delta \colon Q \times \Sigma \to Q$
and a set $A \subseteq Q$ of \emph{accepting states}. We extend
$\delta$ to a function $\hat\delta \colon Q \times \Sigma^* \to Q$ by
setting
$\hat\delta(q,\emptyword) \coloneqq q$
and
$\hat\delta(q,aw) \coloneqq \hat\delta\big(\delta(q,a),w\big)$
 for $a \in \Sigma$ and $w \in \Sigma^*$.
The automaton $\cA$ \emph{accepts} a word $w \in \Sigma^*$ if
$\hat\delta(q_0,w) \in A$. The \emph{language accepted by} $\cA$ is
the set of words accepted by it. A language is called \emph{regular}
if it is accepted by some DFA.

For states $q,q' \in Q$ we say that $q'$ is \emph{reachable} from $q$,
written $q \leadsto q'$, if $\hat\delta(q,w) = q'$ for some word $w
\in \Sigma^*$. If $q'$ is reachable from the initial state $q_0$ we
just call $q'$ reachable. The relation $\leadsto$ is obviously
reflexive and transitive, so the relation
\(
q \approx q' \Leftrightarrow
(q \leadsto q'\text{ and }q' \leadsto q)
\)
is an equivalence relation. Its equivalence classes are called the
\emph{reachability classes} of $\cA$. A reachability class $C$ is
called \emph{closed} if $\delta(q,a) \in C$ for all $q \in C$ and $a
\in \Sigma$%

There are various generalisations of finite automata to
$\omega$-words, resulting in a robust concept of regular
$\omega$-languages, cf.~\cite{Thomas1997}. We use the following
characterisation:
\begin{theorem}[cf.~\cite{Staiger1997}, Thm.~3.2]
  \label{thm:omegareg}
  An $\omega$-language $L \subseteq \Sigma^\omega$ is regular if and only if
  there is a $k \geq 1$ and regular languages $U_i,V_i \subseteq
  \Sigma^*$ for $i = 1,\ldots,k$ such that
  $L = \bigcup_{i=1}^k U_iV_i^\omega$.
\end{theorem}

\subsection{Tree Languages}
\label{sub:treelangs}

We will mostly be concerned with \emph{binary trees} in the sense of
Knuth~\cite[2.3]{taocp1}. These trees have a distinguished root node,
and every node may have a left and/or a right child. We formalise this
as a prefix-closed language $T \subseteq D^\ast$ over the alphabet
$D = \{l ,r\}$. The empty word $\emptyword$ denotes the root of the
tree $T$, a node $u \in T$ may have children $ul, ur \in T$ and is the
\emph{parent} of these.

For a finite alphabet $\Sigma$ and binary tree $T$ we call a function
$\lambda \colon T \to \Sigma $ a \emph{$\Sigma$-labelling} of $T$ and
the pair $(T, \lambda)$ a $\Sigma$-labelled binary tree. We denote
the set of all finite $\Sigma$-labelled binary trees by
$\mathcal{B}^\Sigma$.
We often just write $T \in \mathcal{B}^\Sigma$ and refer to the labelling as $\lambda_T$
when necessary. A set $L \subseteq \mathcal{B}^\Sigma$ is called a
\emph{tree language}. Note that $\mathcal{B}^\Sigma$ contains the
empty tree $\emptyset$.

For $S, T \in \mathcal{B}^\Sigma$ we say that $S$ is a \emph{subtree}
of $T$ and write $S \preceq T$ if there exists $u \in D^\ast$ such
that $uD^\ast \cap T = uS$ and $\lambda_S(v) = \lambda_T(uv)$ for all
nodes $v \in S$. Each node $u \in T$ induces a subtree $T(u)$ of $T$
consisting of $u$ and all its descendants. We write $T_l$ for the
\emph{left subtree} $T(l)$ and $T_r$ for the \emph{right subtree}
$T(r)$ of the root of $T$, respectively.

For two trees $S, T \in \mathcal{B}^\Sigma$ and $a \in \Sigma$ we
define the $\Sigma$-labelled binary tree $\conc_a(S, T)$ which
consists of the root labelled with $a$ and has $S$ and $T$ as left and
right subtrees respectively.

For a tree language $L \subseteq \mathcal{B}^\Sigma$ and
$T \in \mathcal{B}^\Sigma$ we define the tree languages
\[
  \begin{split}
    LT &\coloneqq \bigcup_{a \in \Sigma}\{\conc_a(S,T) \sthat S \in L\}
    \\
    LT^{-1} &\coloneqq \{S \in \mathcal{B}^\Sigma \sthat \text{ There
      exists } a \in \Sigma \text{ with } \conc_a(S,T) \in L\}
  \end{split}
\]
and $TL, T^{-1}L$ are defined analogously. For trees
$T_1,\dotsc,T_k \in \mathcal{B}^\Sigma$ we inductively set
\[
  \begin{split}
    L[T_1]^{-1} &\coloneqq LT_1^{-1} \cup T_1^{-1}L
    \\
    L[T_k,\dotsc,T_1]^{-1} &\coloneqq  (L[T_{k-1},\dotsc,T_1]^{-1})T^{-1}_k \cup T^{-1}_k (L[T_{k-1},\dotsc,T_1]^{-1}),
  \end{split}
\]
so
$L[T_k,\dotsc,T_1]^{-1}$ is the language of all trees that can be
concatenated successively with
$T_k,\dotsc,T_1$ to obtain a tree from $L$.

A \emph{tree automaton} is a tuple $\mathcal{A} = (\Sigma, Q, \Delta,
A)$ consisting of a finite set
$Q$ of \emph{states}, a finite alphabet
$\Sigma$, a \emph{transition relation} $\Delta \subseteq (Q \cup
\{\bot\}) \times (Q \cup \{\bot\}) \times \Sigma \times
Q$, and a set $A \subseteq
Q$ of \emph{accepting} states.  Given a
$\Sigma$-labelled binary tree $T$, a \emph{run} of
$\mathcal{A}$ on $T$ is a function $d \colon D^\ast \to (Q \cup
\{\bot\})$ such that
\begin{itemize}
\item if $u \notin T$ then $d(u) = \bot$, and
\item if $u \in T$ then $(d(ul), d(ur),\lambda_T(u), d(u)) \in \Delta$.
\end{itemize}
A run $d$ is called \emph{accepting} if $d(\emptyword) \in A$ and we say that an automaton $\mathcal{A}$ \textit{accepts} a tree $T$ if there is an accepting run of $\mathcal{A}$ on $T$. 

For states $q \in Q$ and $q' \in Q\cup \{\bot\}$, we say that $q'$ is
$1$-step reachable from $q$ (written $q \leadsto_1 q'$) if
$(q,\tilde q,a,q') \in \Delta$ or $(\tilde q, q, a, q') \in \Delta$
for some $a \in \Sigma$ and $\tilde q \in Q\cup \{\bot\}$. The
reflexive transitive closure $\leadsto \coloneqq \leadsto_1^*$ is called
reachability, and a state $q \in Q$ is called \emph{reachable} if
$\bot \leadsto q$. Equivalently, $q$ is reachable if and only if
there exists a tree $T \in \mathcal{B}^\Sigma$ and a run $d$ of
$\mathcal{A}$ on $T$ with $d(\emptyword) = q$. We call $\mathcal{A}$
\emph{reduced} if every state is reachable.

\subsection{Random Trees and \texorpdfstring{$\omega$-words}{Infinite Words}}

For $n \in \mathbb{N}$ we denote by $\mathcal{B}^\Sigma_n$ and
$\mathcal{B}^\Sigma_{<n}$ the set of $\Sigma$-labelled binary trees of
size $n$ and size strictly less than $n$ respectively.  We consider
$\mathcal{B}^\Sigma_n$ as a finite discrete probability space equipped
with the uniform distribution $\Prob_n$ by setting $\Prob_n(T) =
\frac{1}{\abs{\mathcal{B}_n^\Sigma}}$ for $T \in
\mathcal{B}^\Sigma_n$. Note that $\abs{\mathcal{B}_n^\Sigma} = C_n
\cdot \abs{\Sigma}^n$, where $C_n = \frac{1}{n +1}
\binom{2n}{n} \approx \frac{4^n}{n\sqrt{\pi n}}$ is the $n$-th \emph{Catalan
  number}, cf.~\cite[2.3.4.4]{taocp1}.

The \emph{asymptotic density} (or \emph{asymptotic probability}) of a tree language
$L \subseteq \mathcal{B}^\Sigma$ is defined as
\[
  \Problim(L) \coloneqq \lim_{n \to \infty} \Prob_n(L \cap
  \mathcal{B}^\Sigma_n)
  = \lim_{n \to \infty} \frac{\abs{L \cap \cB^\Sigma_n}}{C_n\cdot \abs{\Sigma}^n},
\]

given the limit exists. We set
$\Problimsup(L) \coloneqq \limsup_{n \to \infty} \Prob_n(L \cap
\mathcal{B}^\Sigma_n)$ and obtain the following lemma, which is easily verified:
\begin{lemma}
  \label{lem:trees:uniform}
  For all $L_1,L_2 \subseteq \mathcal{B}^\Sigma$, the following hold:
  \begin{enumerate}
    \item $\Problimsup(L) = 0 \Leftrightarrow \Problim(L) = 0$
    \item $L_1 \subseteq L_2 \Rightarrow \Problimsup(L_1) \leq \Problimsup(L_2)$
    \item $\Problimsup(L_1 \cup L_2) \leq \Problimsup(L_1) + \Problimsup(L_2)$
    \item $\Problimsup(\mathcal{B}^\Sigma) = 1.$
  \end{enumerate}
\end{lemma}

We use standard terminology from probability theory,
cf.~\cite{Williams1991}. We turn the set $\Sigma^\omega$ of
$\omega$-words over the finite alphabet $\Sigma$ into a probability
space by making each of the projections
$\pi_i \colon \Sigma^\omega \to \Sigma, w_1w_2\ldots \mapsto w_i$
measurable. By $\Prob$ we denote the probability measure for which the
projections are iid random variables with
\( \Prob(w_i = a) = \abs{\Sigma}^{-1} \) for every $i \geq 1$ and
$a \in \Sigma$.
Note that for any $U, V \subseteq \Sigma^*$,
$ UV^\omega = \bigcap_{k \geq 1} \bigcup_{\ell \geq k}
 \big\{ w_1w_2\ldots \sthat w_1\ldots w_\ell \in UV^* \big\}$
is measurable, so by Thm.~\ref{thm:omegareg}, every $\omega$-regular $L
\subseteq \Sigma^\omega$ is measurable and the probability (or
measure) $\Prob(L)$ is well-defined.

We review some basic facts about discrete-time Markov chains with a
finite state space, cf.~\cite{Norris1997}: Fix a finite set $I$ of
\emph{states} and for every $i,j \in I$ a \emph{transition
probability} $p_{ij} \geq 0$ such that $\sum_{j \in I}p_{ij} = 1$ for
every $i \in I$. A \emph{Markov chain} with state space $I$ and
transition probabilities $P = (p_{ij})_{i,j \in I}$ is a sequence
$(X_t)_{t \in \N}$ of random variables taking values in $I$ such that
\(
\Prob(X_{t+1} = j \sthat X_t = i) = p_{ij}
\)
for every $t \geq 0$ and $i,j \in I$. The probability distribution of
$X_0$ is called \emph{initial distribution} of the chain. The initial
distribution and the transition probabilities $P$ together determine
the joint distributions of the $X_t$ by
\[
\Prob(X_0 = i_0, \ldots, X_t = i_t)
= \Prob(X_0 = i_0)\cdot p_{i_0i_1} \cdots p_{i_{t-1}i_t}.
\]
If $\Prob(X_0 = j) = \delta_{ij}$ we say that the chain is started in
state $i$ and denote the resulting probability distribution by
$\Prob_i$. With this definition,
 \( \Prob( X_{s+t} = j \sthat X_s = i ) = \Prob_i(X_t = j)\).
A state $i \in I$ is called \emph{recurrent} if
\(
\Prob_i( X_t = i \text{ for infinitely many t}) = 1.
\)
We say that a state $i \in I$ \emph{leads to} a state $j \in I$,
written $i \to j$, if $\Prob_i(X_t = j\text{ for some }t \in \N) > 0$,
and that $i$ and $j$ \emph{communicate} (written $i \leftrightarrow
j$) if both $i\to j$ and $j\to i$. Then $\leftrightarrow$ is an
equivalence relation on $I$ and its equivalence classes are just
called \emph{classes} of states. A class $C \subseteq I$ is called
\emph{closed} if $i \in C$ and $i \to j$ imply $j \in C$. We need the
following theorem:
\begin{theorem}[{cf.~\cite[Thm.~1.5.6]{Norris1997}}]
  \label{thm:recurrent}
  If $C \subseteq I$ is a closed class, every $i \in C$ is recurrent.
\end{theorem}
\begin{corollary}
  \label{cor:markovrecurr}
  If $C \subseteq I$ is a closed class and $i \in C$, then
  $\Prob(X_t = i\text{ infinitely often} \sthat X_t \in C\text{ for some
  }t) = 1$ for every $s \in \N$.
\end{corollary}
\section{Characterising Sparseness of Regular Tree Languages}
\label{sec:trees}

In this section we exactly characterise regular tree languages with
asymptotic density $0$ by excluded factors, namely \emph{forbidden
  subtrees}. This generalises Sin'ya's result for regular languages.
The well-known \emph{infinite monkey theorem} states that a language
of finite words $L \subseteq \Sigma^*$ has asymptotic density $1$ if
$\Sigma^* x \Sigma^* \subseteq L$ for some $x \in \Sigma^*$. This has
been generalised to tree languages by Asada et
al.~\cite[Thm.~2.13]{Asada2019},
who prove that \emph{contexts} of up to logarithmic size appear
asymptotically almost surely in certain regular tree languages.
We only need a weaker version stated in
Theorem~\ref{thm:monkey:trees}, and give a comparatively short proof of
it using methods from \emph{analytic combinatorics}~\cite{Flajolet2009}. 
In~\Cref{thm:tree:condition} we then show that, for a
\emph{regular} tree language, the existence of a forbidden subtree is
a necessary condition for sparseness.

\begin{theorem}
  \label{thm:monkey:trees}
  $\Problim(\{T \in \mathcal{B}^{\Sigma} \sthat S \preceq T\}) = 1$
  for every nonempty $S \in \mathcal{B}^\Sigma$.
\end{theorem}
\begin{proof}
 First, note that 
 \[
 \begin{split}
   \Problimsup(\{T \in \mathcal{B}^{\Sigma} \sthat S \preceq T\})
   &=
   \limsup (1- \Prob \{T \in \mathcal{B}^{\Sigma} \sthat S \npreceq
   T\})
   \\
   &=
   1 - \liminf_{n \to \infty}\frac{\abs{\{T \in
       \mathcal{B}^\Sigma_n\sthat S \npreceq
       T\}}}{C_n\abs{\Sigma}^n}.
 \end{split}
 \]
  We fix a nonempty tree $S \in \mathcal{B}^\Sigma$ and examine the asymptotic
  behaviour of the sequence
  \[
  a_n \coloneqq \abs{\{T \in \mathcal{B}^\Sigma_n\sthat S \npreceq
    T\}}
  \]
  following the approach of~\cite[Example
    III.41]{Flajolet2009}.
  To analyse the generating function of the sequence $(a_n)_n$, 
  we denote by $f_{n,k}$ the number of $\Sigma$-labelled binary trees
  of size $n$ that contain $S$ as a subtree at $k$ different
  positions, and let
  \[
  f(u,z) \coloneqq \sum_{n,k \geq 0} f_{n,k}z^n u^k
  \]
  be its bivariate generating function. In particular $a_n = f_{n,0}$
  and $f(0,z) = \sum_n a_n z^n$ is the generating function of the
  sequence $(a_n)_{n \geq 0}$.

  By $\omega(T)$ we denote the number of distinct occurrences of $S$
  in a tree $T \in \mathcal{B}^\Sigma$. Then $\omega(\emptyset) = 0$
  and since $S$ can occur in the left or right subtree, or be the
  whole tree $T$, we get
  \(
  \omega(T) = \omega(T_l) + \omega(T_r) + [T = S]
  \),
  where $[T = S]$ is $1$ if $T = S$ and $0$ otherwise,
  for $T \not= \emptyset$.
  For the function $u \mapsto u^{\omega(T)}$ this can be
  rewritten as
  \begin{equation}
    \label{eqn:occurrences}
    u^{\omega(T)} = u^{\omega(T_l)} u^{\omega(T_r)} u^{[T = S]} = u^{\omega(T_l)} u^{\omega(T_r)} + [T = S](u-1),
  \end{equation}
  and justifying algebraic
  manipulations of formal power series as in~\cite[A.5.]{Flajolet2009}, we get:
  \begin{align*}
    f(u, z) & = \sum_{n=0}^\infty z^n \sum_{k=0}^\infty u^k f_{n,k} = \sum_{n=0}^\infty z^n \sum_{T \in \mathcal{B}^\Sigma_n} u^{\omega(T)}                                    \\
            & \overset{\eqref{eqn:occurrences}}{=} \sum_{n=0}^\infty z^n \sum_{T \in \mathcal{B}^\Sigma_n} \Big([T = S](u-1) + u^{\omega(T_l)} u^{\omega(T_r)} \Big)                     \\
            & = \sum_{n=0}^\infty z^n \sum_{T \in \mathcal{B}^\Sigma_n} [T = S](u-1) + \sum_{n=0}^\infty z^n \sum_{T \in \mathcal{B}^\Sigma_n} u^{\omega(T_l)} u^{\omega(T_r)} \\
            & = z^{m}(u-1) + 1 + \sum_{n=1}^\infty z^n \sum_{T \in \mathcal{B}^\Sigma_n} u^{\omega(T_l)} u^{\omega(T_r)}
  \end{align*}
  for $m \coloneqq \abs{S}$.
  We set
  \(
    f_n(u) \coloneqq \sum_{T \in \mathcal{B}^\Sigma_n} u^{\omega(T)} =  \sum_{k=0}^\infty u^k f_{n,k}
  \)
  and get
  \begin{align*}
    f(u, z) - z^{m}(u-1) - 1 & = \sum_{n=1}^\infty z^n \sum_{T \in \mathcal{B}^\Sigma_n} u^{\omega(T_l)} u^{\omega(T_r)}                                                                                                        \\
                             & = \sum_{n=1}^\infty z^n \sum_{j=0}^{n-1} \abs{\Sigma} \Bigl(\sum_{T_l \in \mathcal{B}^\Sigma_j}u^{\omega(T_l)}\Bigr) \Bigl(\sum_{T_r \in \mathcal{B}^\Sigma_{n-1-j}} u^{\omega(T_r)}\Bigr)       \\
                             & = z \abs{\Sigma} \sum_{n=1}^\infty z^{n-1} \sum_{j=0}^{n-1} \Bigl(\sum_{T_l \in \mathcal{B}^\Sigma_j}u^{\omega(T_l)}\Bigr) \Bigl(\sum_{T_r \in \mathcal{B}^\Sigma_{n-1-j}} u^{\omega(T_r)}\Bigr) \\
                             & = z \abs{\Sigma} \sum_{n=0}^\infty
    \sum_{j=0}^{n} z^j f_j(u) \cdot z^{n-j}f_{n-j}(u)                                                                                                                        
                              = z \abs{\Sigma} f(u,z)^2.
  \end{align*}
  Solving this quadratic equation for $f(u,z)$ gives two candidate solutions
  \[
  f(u, z) = \frac{1 \pm \sqrt{1-4z\abs{\Sigma}-4\abs{\Sigma}z^{m+1}(u-1)}}{2z\abs{\Sigma}},
  \]
  and since $f(1,\frac{z}{\abs{\Sigma}})$ is the generating function
  of the Catalan numbers, subtracting the square root gives the right solution.
  The generating function $f(0 ,z)$ of the sequence $(a_n)_{n\geq 0}$
  is now given by
  \[
    f(0, z) = \frac{1-\sqrt{1-4z\abs{\Sigma}+4\abs{\Sigma}z^{m+1}}}{2z\abs{\Sigma}}.
  \]
  The function $f(0,z)$ is analytic at $0$ by extending it to
  $f(0,0)=1$.  The radius around $0$, where $f(0,z)$ is analytic is
  exactly the radius $R$, for which the polynomial
  $p(z)\coloneqq 1-4\abs{\Sigma}z + 4\abs{\Sigma}z^{m+1}$ is non-zero
  (cf. analyticity of $\sqrt{1-z}$).  Considering the reciprocal
  polynomial of $p$ yields~$R > \frac{1}{4 \abs{\Sigma}}$.

  Finally, by~\cite[Theorem IV.7 (Exponential Growth
  Formula)]{Flajolet2009} there exist a subexponential factor
  $(\eta_n)_n$ such that $a_n = R^{-n} \eta_n$.  Also, by Stirling's
  formula there exists a subexponential factor $(\theta_n)_n$ such
  that $C_n = 4^{-n} \theta_n$. In total, this yields
  \[
    \frac{\abs{\{T \in \mathcal{B}^\Sigma_n\sthat S \npreceq T\}}}{\abs{\mathcal{B}^\Sigma_n}} = \frac{a_n}{C_n \abs{\Sigma}^n}
    = \frac{R^{-n} \eta_n}{4^n \theta_n \abs{\Sigma}^n}
    = \Bigl(\frac{1}{4 R \abs{\Sigma}}\Bigr)^n \frac{\eta_n}{\theta_n}.
  \]
  Since the factors $\eta_n$ and $\theta_n$ are subexponential the
  sequence converges to $0$ as $n$ tends to infinity.
\end{proof}

In order to show the converse for \emph{regular} tree languages, we
lift the proof from~\cite{Koga2019} to the case of tree languages.
First, we derive bounds for the asymptotic density of specific tree
languages.

\begin{lemma} \label{lem:embed:prob}
  For $L \subseteq \mathcal{B}^\Sigma$ and
  $T, T_1,\dotsc,T_k \in \mathcal{B}^\Sigma$ the following hold:
  \begin{enumerate}
    \item $\Problimsup(LT) = \frac{1}{\abs{\Sigma}^{\abs{T}} 4^{\abs{T}+1}} \Problimsup(L)$
    \item $\Problimsup(LT^{-1}), \Problimsup(T^{-1}L) \leq \abs{\Sigma}^{\abs{T}} 4^{\abs{T}+1} \Problimsup(L)$
    \item $\Problimsup(L [T_k, \dotsc, T_1]^{-1}) \leq  2^k \abs{\Sigma}^{\sum_i\abs{T_i}} 4^{\sum_i\abs{T_i}+k} \Problimsup(L)$
  \end{enumerate}
\end{lemma}
\begin{proof}
  By
  \[
    \begin{split}
      \Prob_n(LT \cap \mathcal{B}^\Sigma_n) & = \frac{\abs{(L \cap
          \mathcal{B}^\Sigma_{n-\abs{T}-1})T}}{\abs{\mathcal{B}^\Sigma_n}}
      = \frac{\abs{L \cap
          \mathcal{B}^\Sigma_{n-\abs{T}-1}}}{\abs{\mathcal{B}^\Sigma_{n-\abs{T}-1}}}
      \frac{\abs{\Sigma} \cdot
        \abs{\mathcal{B}^\Sigma_{n-\abs{T}-1}}}{\abs{\mathcal{B}^\Sigma_n}}
      \\ & = \Prob_{n-\abs{T}-1}(L \cap
      \mathcal{B}^\Sigma_{n-\abs{T}-1}) \frac{C_{n-\abs{T}-1}
      \abs{\Sigma}^{n-\abs{T}}}{C_n \abs{\Sigma}^n}                \\ & =
      \Prob_{n-\abs{T}-1}(L \cap \mathcal{B}^\Sigma_{n-\abs{T}-1})
      \frac{1}{\abs{\Sigma}^{\abs{T}}} \frac{C_{n-\abs{T}-1}}{C_n},
    \end{split}
  \]
  taking the limes superior on both sides together with the identity
  $\lim_{n \to \infty} \frac{C_{n-k}}{C_n} = \frac{1}{4^k}$ proves the
  first part.  For the second part note that $(LT^{-1})T \subseteq L$
  and hence by \Cref{lem:trees:uniform} together with the first part we obtain
  \[
    \Problimsup(L) \geq \Problimsup\big((LT^{-1})T\big) = \frac{1}{\abs{\Sigma}^{\abs{T}} 4^{\abs{T}+1}} \Problimsup(LT^{-1}),
  \]
  and likewise for $\Problimsup(T^{-1}L)$.
   Finally, for the third part we have
   \begin{align*}
     & \Problimsup(L [T_k, \dotsc, T_1]^{-1}) \\ & = \Problimsup\Bigl((L[T_{k-1},\dotsc,T_1]^{-1})T^{-1}_k \cup T^{-1}_k (L[T_{k-1},\dotsc,T_1]^{-1})\Bigr) \\
     & \leq \Problimsup\Bigl((L[T_{k-1},\dotsc,T_1]^{-1})T^{-1}_k\Bigr) + \Problimsup\Bigl((L[T_{k-1},\dotsc,T_1]^{-1})T^{-1}_k\Bigr) \\
     & \leq 2 \abs{\Sigma}^{\abs{T_k}} 4^{\abs{T_k}+1} \Problimsup(L [T_{k-1}, \dotsc, T_1]^{-1})
   \end{align*}
   by using the second part for the last inequality. 
   Hence, the claim follows by induction. 
\end{proof}

\begin{lemma} \label{lem:reachable} Let
  $\mathcal{A} = (Q, \Sigma, \Delta, A)$ be a tree automaton.  For
  every reachable state $q \in Q$ there exists a tree
  $T \in \mathcal{B}^\Sigma$ with $\abs{T} \leq 2^{\abs{Q}} - 1$ such
  that a run of $\mathcal{A}$ on $T$ ends in $q$.
\end{lemma}
\begin{proof}
  Let $q \in Q$ be reachable by $\mathcal{A}$, so there exists
  $T \in \mathcal{B}^\Sigma$ such that a run of $\mathcal{A}$ on $T$
  ends in $q$.  If $\abs{T} \leq 2^{\abs{Q}} - 1$ we are done, so
  assume $\abs{T} > 2^{\abs{Q}} - 1$.  Then there exists $w \in T$
  with $\abs{w} \geq \abs{Q}$.  Let
  $d\colon D^\ast \to Q \cup \{\bot\}$ be a run of $\mathcal{A}$ on
  $T$, so for all $u \in T$ it holds
  $(d(ul),d(ur), \lambda(u), d(u)) \in \Delta$.  We obtain a sequence
  $[d(p^n(w))]_{n=0}^{\abs{w}}$ of states, which $\mathcal{A}$ passes
  from $w$ to $p^{\abs{w}}(w)=\emptyword$.  Since
  $\abs{w} \geq \abs{Q}$, by the pigeonhole principle, there must be a
  state in $Q$ which occurs twice in $[d(p^n(w))]_{n=0}^{\abs{w}}$.
  Let $i,j \in \{0,\dotsc,\abs{w}\}$ with $i \neq j$ and
  $d(p^i(w)) = d(p^j(w))$, then either $T(p^i(w)) \preceq T(p^j(w))$
  or $T(p^j(w)) \preceq T(p^i(w))$.  We assume
  $T(p^j(w)) \preceq T(p^i(w))$, so $j>i$ (the other case is
  analogous).  We obtain a new tree $T$ by replacing $T(p^i(w))$ by
  $T(p^j(w))$ in $T$ and set
  $w^1 \coloneqq w_0 \dotsc w_{i-1}w_j \dotsc w_{\abs{w}+1}$, which is
  the new node at the position of $w$.  Then
  $\abs{w^1} = \abs{w}-\abs{i-j} < \abs{w}$.  Since $\mathcal{A}$ ends
  in the same state after running on $T(p^j(w))$ and $T(p^i(w))$, it
  still ends in $q$ after running on $T_1$.  If
  $\abs{w^1} \geq \abs{Q}$ the same argument applies for the sequence
  $[d(w),d(p(w)),\dotsc,d(p^j(w)),d(p^{i+1}(w)),\dotsc,d(\emptyword)]$
  and we iteratively obtain $\abs{w^\ell} < \abs{Q}$ after at most
  $\ell$ iterations.  The same argument can be applied to every node
  $v \in T_\ell$ with $\abs{v} \geq \abs{Q}$ until for all nodes $v$
  in the resulting tree it holds $\abs{v} < \abs{Q}$.  A binary tree
  with this property has depth at most $\abs{Q}-1$, so its size is
  bounded by $\sum_{k=0}^{\abs{Q}-1}2^k = 2^{\abs{Q}}-1$.
\end{proof}

Next, we show that if a binary tree $S$ occurs as a subtree in a regular
language $L$, then there is a tree $T \in L$ with $S \preceq T$ that
is 'not much larger' than $S$:

\begin{lemma}
  \label{lem:bound}
  For every regular tree language $L \subseteq \mathcal{B}^\Sigma$
  there exists $n \in \mathbb{N}$ such that for all
  $S \in \mathcal{B}^\Sigma$ and $T \in L$ with $S \preceq T$ there
  exist $T_1, \dotsc, T_k \in \mathcal{B}^\Sigma_{< 2^n}$ with
  $k \leq n$ such that $S \in L [T_k, \dotsc, T_1]^{-1}$.
\end{lemma}
\begin{proof}
  Let $\mathcal{A} = (\Sigma, Q, \Delta, A)$ be a tree automaton
  recognising the language $L$.  For $T \in L$ and
  $S \in \mathcal{B}^\Sigma$ with $S \preceq T$ there exists
  $w = w_1\cdots w_\ell \in T$ with $S = T(w)$. Let
  $d\colon D^\ast \to Q \cup \{\bot\}$ be a run of $\mathcal{A}$ on
  $T$ and $d_i = d(w_1\cdots w_i) \in Q$ be state of $\mathcal{A}$ at
  node $w_1\cdots w_i$ in this run, for $i = 0,\ldots,\ell$. Then
  $d_0 \in A$ (because $\mathcal{A}$ accepts $T$) and
  $d_{i} \leadsto_1 d_{i-1}$ for $i = 1,\ldots,\ell$. If $d_i = d_j$ for
  some $i < j$ we remove that subsequence $d_i,\ldots,d_{j-1}$ and
  repeat this process until we get a sequence $d'_0,\ldots,d'_{m}$ with
  $m < \abs{Q}$.

  Note that $d_S : D^\ast \to Q\cup\{\bot\}$ with $d_S(v) \coloneqq d(wv)$ is
  a run of $\mathcal{A}$ on $S$ with
  $d_S(\epsilon) = d(w) = d_\ell = d'_n$, and our reduced sequence
  $d'$ still satisfies $d'_i \leadsto_1 d'_{i-1}$ for
  $i = 1,\ldots,m$. The result, with $n = \abs{Q}$, now follows from
  Lemma~\ref{lem:reachable} and the definition of $\leadsto_1$.
\end{proof}

We use the previous lemma to show that if a regular tree language $L$
is sparse, then it must already admit a \emph{forbidden subtree}.
That is, a fixed tree $S$ which does not occur as subtree of any tree in $L$.

\begin{theorem}
  \label{thm:koga:trees}
  Let $L \subseteq \mathcal{B}^\Sigma$ be a regular tree language with
  $\Problim(L) = 0$.  Then there exists $S \in \mathcal{B}^\Sigma$
  such that
  $\{T \in \mathcal{B}^\Sigma \sthat S \preceq T\} \cap L =
    \emptyset$.
\end{theorem}
\begin{proof}
  We argue by contraposition and assume that for all
  $S \in \mathcal{B}^\Sigma$ we have
  $\{T \in \mathcal{B}^\Sigma \sthat S \preceq T\} \cap L \neq
    \emptyset$.  That is, for every $S \in \mathcal{B}^\Sigma$ there
  exists $T \in L$ with $S \preceq T$.  By~\Cref{lem:bound}, we infer
  that for every $S \in \mathcal{B}^\Sigma$ there exist $k \leq n$ and
  $T_1,\dotsc,T_k \in \mathcal{B}^\Sigma_{<2^n}$ such that
  $S \in L [T_k, \dotsc, T_1]^{-1}$.  This in turn is equivalent to
  \[
    \mathcal{B}^\Sigma \subseteq \bigcup_{k=1}^n \bigcup_{T_1,\dotsc,T_k \in \mathcal{B}^\Sigma_{<2^n}} L [T_k, \dotsc, T_1]^{-1}.
  \]
  Using~\Cref{lem:trees:uniform} and $k$ successive applications
  of~\Cref{lem:embed:prob} we obtain:
  \begin{align*}
    1 = \Problimsup(\mathcal{B}^\Sigma) & \leq \Problimsup\Bigl(\bigcup_{k=1}^{n} \bigcup_{T_1,\dotsc,T_k \in \mathcal{B}^\Sigma_{<2^n}} L [T_k, \dotsc, T_1]^{-1}\Bigr)                   \\
                                        & \leq \sum_{k=1}^n \sum_{T_1,\dotsc,T_k \in \mathcal{B}^\Sigma_{< 2^n}} \Problimsup(L [T_k, \dotsc, T_1]^{-1})                                    \\
                                        & \leq \sum_{k=1}^n
                                          \sum_{T_1,\dotsc,T_k \in
                                          \mathcal{B}^\Sigma_{< 2^n}}
                                          2^k \prod_{j=1}^k
                                          \abs{\Sigma}^{\abs{T_j}}
                                          4^{\abs{T_j}+1}
                                          \Problimsup(L)
  \end{align*}
  Thus, we conclude that $\Problim(L) > 0$.
\end{proof}

This converse of the infinite monkey theorem for regular tree languages
provides a characterisation of sparseness in terms of
\emph{forbidden subtrees}.  In order to also obtain
such a characterisation in terms of tree automata akin
to~\cite{Sinya2015}, we give the following definition.

\begin{definition}
  Let $\mathcal{A} = (\Sigma, Q, \Delta, A)$ be a tree automaton. A
  set of states $V \subseteq Q$ is called a \emph{sink} if for all
  $q \in V$ and all $q_l,q_r \in Q \cup \{\bot\}, a \in \Sigma$ it
  holds that $\delta(q_l,q,a), \delta(q,q_r,a) \subseteq V$. That is,
  $V$ is a sink exactly if every run of $\mathcal{A}$ remains in $V$
  once it entered a state in $V$.
\end{definition}

Finally, we conclude the following characterisation of sparse 
regular tree languages.
\begin{theorem}
  \label{thm:tree:condition}
  Let $L$ be a regular tree language and
  $\mathcal{A}= (\Sigma, Q, \Delta, A)$ be a reduced tree automaton
  recognising $L$. Then the following assertions are equivalent:
  \begin{enumerate}
    \item $\Problim(L) = 0$ \label{condition:zero}
    \item There exists a tree $S\in \mathcal{B}^\Sigma$ such that
          $\{T \in \mathcal{B}^\Sigma\sthat S \preceq T\} \subseteq
            \mathcal{B}^\Sigma \setminus L$ \label{condition:subtree}
    \item $\mathcal{A}$ has a sink $V \subseteq Q$ with
          $V \cap A = \emptyset$ \label{condition:sink}
  \end{enumerate}
\end{theorem}
\begin{proof}
  Theorems~\ref{thm:monkey:trees} and~\ref{thm:koga:trees} together
  imply \ref{condition:zero} $\Leftrightarrow$ \ref{condition:subtree}.

  \ref{condition:subtree} $\Rightarrow$ \ref{condition:sink}: Let
  $S \in \mathcal{B}^\Sigma$ such that
  $\{T \in \mathcal{B}^\Sigma\sthat S \preceq T\} \subseteq
    \mathcal{B}^\Sigma \setminus L$.  Let $q_0 \in Q$ be a state in
  which $\mathcal{A}$ ends after reading $S$. Then
  $q_0 \in Q \setminus A$ since $\mathcal{A}$ has to reject $S$
  because all trees which contain $S$ as subtree have to be rejected.
  Therefore, all $q_l,q_r \in Q \cup \{\bot\}$ and $a \in \Sigma$
  satisfy $\delta(q_l,q_0,a) \subseteq Q \setminus A$ and
  $\delta(q_0,q_r,a) \subseteq Q \setminus A$ because otherwise one
  could construct a tree (since every state is reachable) which
  contains $S$ as subtree and is accepted by $\mathcal{A}$.  Hence,
  there exists a sink $S \subseteq Q \setminus A$ which contains
  $q_0$.

  \ref{condition:sink} $\Rightarrow$ \ref{condition:subtree}: Let
  $V \subseteq Q \setminus A$ be a sink of $\mathcal{A}$. There exists
  a tree $S$ on which $\mathcal{A}$ ends in a state of $V$.  If
  $\mathcal{A}$ runs on any tree containing $S$ as a subtree,
  $\mathcal{A}$ cannot leave the sink and thus cannot leave
  $Q \setminus A$.  This yields
  $\{T \in \mathcal{B}^\Sigma\sthat S \preceq T\} \subseteq L$.
\end{proof}
\begin{remark}
  By duality,~\Cref{thm:tree:condition} also provides a
  characterisation of regular tree languages $L$ with asymptotic
  density $1$: It holds $\Problim(L)=1$ if and only if there exists a tree
  $S$ such that $S \preceq T$ implies $T \in L$, and this is the
  case if and only if an automaton recognising $L$ has a sink
  consisting of accepting states.
\end{remark}

As a consequence of this characterisation, we obtain a simple linear
time algorithm for deciding sparseness (or denseness) of regular tree
languages akin to the algorithm in \cite{Sinya2015}.  Note that in
contrast to the characterisation given there, here we do not need to
require that a sink is a strongly connected component. If there is a
sink $V \subseteq Q$ with $V \subseteq Q \setminus A$, there cannot be
a sink $V' \subseteq A$.

\begin{corollary}
  \label{cor:trees:decidable}
  Let $L \subseteq \mathcal{B}^\Sigma$ be a regular tree language.
  There is an algorithm deciding whether $L$ has asymptotic density $0$ or $1$
  in time $O(n)$, where $n$ is the number of states
  of a given deterministic tree automaton $\mathcal{A}$ recognising $L$.
\end{corollary}
\begin{proof}
  We define the \emph{state graph} of a deterministic tree automaton $\mathcal{A}=(Q, \Sigma, \delta, A)$ as
  $G_{\mathcal{A}} = (Q, \{(q,q') \sthat \exists s \in Q, a \in \Sigma. \, \delta(s,q,a) = q' \text{ or } \delta(q,s,a) = q'\})$.
  For a set $V \subseteq Q$ we define $N^+_{G_{\mathcal{A}}}(V) \coloneqq \{q' \in Q \setminus V: \exists q \in V. (q,q') \in E(G_{\mathcal{A}})\}$.
  A \emph{strongly connected component} of $\mathcal{A}$ is a strongly connected component of the directed graph $G_{\mathcal{A}}$.
  The linear time algorithm now is as follows: 
  \begin{enumerate}
    \item Compute the set of strongly connected components of $G_{\mathcal{A}}$.
    \item For each strongly connected component $V \subseteq Q$, check whether
    \begin{enumerate}[label*=\arabic*.]
      \item $N^+_{G_{\mathcal{A}}}(V) = \emptyset$ and
      \item $V \subseteq A$ or $V \subseteq Q \setminus A$.
    \end{enumerate}
  \end{enumerate}
  For the correctness, we observe that $\mathcal{A}$ has a sink if and only if there exists 
  a strongly connected component $V \subseteq Q$ with $N^+_{G_{\mathcal{A}}}(V) = \emptyset$. 
  Then the asymptotic density of $L$ is determined by checking $V \subseteq A$ or 
  $V \subseteq Q \setminus A$ by \Cref{thm:tree:condition}.
  For the running time, the first step can be implemented to run in time $O(n + n\abs{\Sigma}) = O(n)$ by \cite{Tarjan1972}.
  Afterwards, for each of the at most $n$ strongly connected components only constant-time accesses to the adjacency of $G_{\mathcal{A}}$
  and accepting states of $\mathcal{A}$ are necessary. 
\end{proof}

\paragraph*{Unranked Trees.}

There is a well-known correspondence between binary trees in our sense
(with left and right children) and forests of unranked trees, see, for
example, Section~2.3.2 of~\cite{taocp1}. Every unranked tree $T$ can
be uniquely encoded by a binary tree $T^\flat$, and the binary trees
$T^\flat$ obtained in this way are exactly those in which the root
does not have a right child (i.e.~$T^\flat_r = \emptyset$). Again
there are various approaches to defining regular languages of unranked
trees (such as definability in monadic second-order logic), all of
which are equivalent to saying that a set $L$ of unranked trees is
regular if the language
\[
  L^\flat \coloneqq \{ T^\flat \sthat T \in L \}
\]
of binary trees is regular in our
sense. Theorem~\ref{thm:tree:condition} therefore gives an exact and
decidable characterisation of regular languages of unranked trees.
However, this gives a necessary and sufficient condition on $L^\flat$
for when a regular language $L$ is sparse. We now show that in fact:
\begin{theorem}
  A regular language $L$ of unranked trees is sparse if, and only if,
  some unranked tree $S$ does not appear as a subtree of any tree $T
  \in L$.
\end{theorem}
\begin{proof}
  Let $S$ be an unranked tree with root label
  $a \coloneqq \lambda_S(\epsilon) \in \Sigma$. Then $S$ is a subtree of an
  unranked tree $T$ if, and only if, $T^\flat$ has a node labelled $a$
  whose left subtree is exactly $S^\flat_l$, the left subtree of the
  root of $S^\flat$. (Note that $S^\flat$, being the encoding of an
  unranked tree, has $S^\flat_r = \emptyset$.) Let us say that
  $S^\flat_l$ is an $a$-left subtree of $T^\flat$ in this situation.

  Now if $\conc_a(S^\flat_l,\emptyset) \preceq T$ then in particular
  $S^\flat_l$ is an $a$-left subtree of $T$, and with
  Theorem~\ref{thm:monkey:trees} we get that
  \[
    \Problim(\{ T^\flat \in \mathcal{B} \sthat S^\flat_l\text{ is an $a$-left subtreee of }T^\flat \}) = 1
  \]
  for every $S^\flat_l \in \mathcal{B}$. Similarly, it is easy to
  adapt the statements and proofs of Lemma~\ref{lem:bound} and
  Theorem~\ref{thm:koga:trees} to the case that $S^\flat_l$ is an
  $a$-left subtree of $T^\flat$.
\end{proof}

\paragraph*{Proving Non-Regularity.}

In another direction,~\Cref{thm:tree:condition} shows a sufficient condition for
the non-regularity of tree languages:

\begin{corollary}
  \label{cor:non-regular}
  Let $L\subseteq \mathcal{B}^\Sigma$ be a tree language with
  $\Problim(L) = 0$.  If $L$ does not have a forbidden subtree, then
  $L$ is not regular.
\end{corollary}

\begin{example}
  We call a binary tree $T \in \mathcal{B}^\Sigma$ \emph{symmetric} if
  the left and right subtree obtained from the root are isomorphic.  Let
  $L_{\text{sym}}$ be the set of all symmetric binary trees, then for
  every $a \in \Sigma$ and $S \in \mathcal{B}^\Sigma$ we have
  $\conc_a(S,S) \in L_{\text{sym}}$ and thus $L_{\text{sym}}$ does not
  have a forbidden subtree.  However, for the asymptotic density of
  $L_{\text{sym}}$ we get $\abs{L_{\text{sym}} \cap
      \mathcal{B}^\Sigma_{2n}} = 0$ and
  \[
    \lim_{n \to \infty} \frac{\abs{L_{\text{sym}} \cap \mathcal{B}^\Sigma_{2n+1}}}{\abs{\mathcal{B}^\Sigma_{2n+1}}} = \lim_{n \to \infty} \frac{C_n \abs{\Sigma}^{n+1}}{C_{2n+1} \abs{\Sigma}^{2n+1}} = \lim_{n \to \infty} \frac{C_n}{C_{2n+1} \abs{\Sigma}^n} = 0,
  \]
  which yield $\Problim(L_{\text{sym}}) =
    0$. By~\Cref{cor:non-regular} we infer that $L_{\text{sym}}$ is not
  a regular tree language.
\end{example}

\paragraph*{Random Binary Search Trees.}

Another prominent model for random binary trees are \emph{random binary search trees}.
These trees appear in the analysis of algorithms such as
Quicksort and Find, cf.~\cite{Devroye1986}.
A random binary search tree on $n$ nodes is obtained by taking a root and appending to
it a left subtree of size $k$ and a right subtree of size $n-1-k$
independently, where $k$ is chosen uniformly at random from
$\{0,\dotsc,n-1\}$.  Formally, we let
$\Prob^{\text{bst}}_1: \mathcal{B}^\Sigma_1 \to [0,1], ~ T \mapsto
\frac{1}{\abs{\Sigma}}$ and for $n>1$ and $T \in \mathcal{B}^\Sigma_n$
set
$\Prob^{\text{bst}}_n(T) = \frac{1}{n
  \abs{\Sigma}}\Prob^{\text{bst}}_{\abs{T_l}}(T_l)
\Prob^{\text{bst}}_{\abs{T_r}}(T_r)$.  For the asymptotic probability
of a language $L \subseteq \mathcal{B}^\Sigma$ we again set
\(
\Problim^{\text{bst}}(L) \coloneqq \lim_{n \to \infty}
  \Prob^{\text{bst}}_n(L),
\)
given the limit exists.

The characterisation from~\Cref{thm:tree:condition} however does not
immediately hold for non-uniform probability measures.  In the
following, we consider random binary search trees and show that there
are tree languages with asymptotic probability $0$ which do not admit
our characterisation.
\begin{lemma}
  \label{lem:bst}
  Let $L \subseteq \mathcal{B}^\Sigma$ and $T \in
  \mathcal{B}^\Sigma$. Then $\Problim^{\text{bst}}(LT) = 0 $.
\end{lemma}
\begin{proof}
  We use the independence condition in the definition of the binary search tree distribution to obtain the following:
  \begin{align*}
    \Prob^{\text{bst}}_n(LT \cap \mathcal{B}^\Sigma_n) &
                                                         = \sum_{S \in LT \cap \mathcal{B}^\Sigma_n} \Prob^{\text{bst}}_n(S)                                                                                                                                                 \\
                                                       & = \sum_{S \in L \cap \mathcal{B}^\Sigma_{n-1-\abs{T}}} \sum_{a \in \Sigma} \Prob^{\text{bst}}_n(\conc_a(S,T))                                                  \\
                                                       & = \sum_{S \in L \cap \mathcal{B}^\Sigma_{n-1-\abs{T}}} \sum_{a \in \Sigma} \frac{1}{n |\Sigma|}\Prob^{\text{bst}}_{\abs{T}}(T) \Prob^{\text{bst}}_{\abs{S}}(S) \\
                                                       & = \sum_{S \in L \cap \mathcal{B}^\Sigma_{n-1-\abs{T}}} \frac{1}{n} \Prob^{\text{bst}}_{\abs{T}}(T) \Prob^{\text{bst}}_{\abs{S}}(S)                             \\
                                                       & = \frac{1}{n} \Prob^{\text{bst}}_{\abs{T}}(T) \sum_{S \in L \cap \mathcal{B}^\Sigma_{n-1-\abs{T}}} \Prob^{\text{bst}}_{\abs{S}}(S) \\
                                                       & = \frac{1}{n} \Prob^{\text{bst}}_{\abs{T}}(T) \Prob^{\text{bst}}_{n-1-\abs{T}}(L \cap \mathcal{B}^\Sigma_{n-1-\abs{T}}) \leq \frac{1}{n}.
  \end{align*}
  Taking the limit yields the desired result.
\end{proof}

\begin{example}
  \label{ex:bst}
  Consider the tree language $R$ that consists of all
  $\Sigma$-labelled binary trees $T$ with $T_l = \{\emptyword\}$, i.e.,
  empty left subtree.  The language $R$ is regular because the automaton
  $\mathcal{A}= (\{q_0,q_1,q_2\},\{a\},\Delta,\{q_2\})$ with
  \[ \Delta \coloneqq \{(\bot,\bot,a,q_0)\} \cup \{(l,r,a,q_1)\sthat r \neq q_0, (l,r) \neq (\bot,\bot)\} \cup \{(l,q_0,a,q_2) \sthat l \in Q \cup \{\bot\}\} \]
  recognises $R$.
  
  By~\Cref{lem:bst} we have $\Problim^{\text{bst}}(R) = 0$, but on the
  other hand, for all $S \in \mathcal{B}^\Sigma$ it holds that 
  $\{T \in \mathcal{B}^\Sigma\sthat S \preceq T\} \not\subseteq
  \mathcal{B}^\Sigma \setminus L$ since every tree might occur as a
  subtree in a right subtree from $R$. Also, the automaton $\mathcal{A}$
  does not have a sink.
\end{example}

\section{Infinite Words}
\label{sec:infwords}

We first review Sin'ya's result and Koga's simplified
proof~\cite{Koga2019} of it and then see how it can be extended to
infinite words.

\begin{definition}
  \label{def:infixlang}
  For a language $L \subseteq \Sigma^*$, the \emph{infix language}
  $\infix(L)$ is defined as
  \begin{equation}
    \label{eqn:infixlang}
    \infix(L) \coloneqq \{ w \in \Sigma^* \sthat xwy \in L\text{ for some }x,y \in
    \Sigma^* \}.
  \end{equation}
  The language is said to be \emph{infix complete} if $\infix(L) =
  \Sigma^*$.
\end{definition}

In~\cite{Sinya2015}, Sin'ya proved that a regular language $L$ has
asymptotic density strictly larger than $0$ if and only if it is infix
complete. Koga's simplified proof in~\cite{Koga2019} hinges on the
fact that for regular languages, the length of the prefix $x$ and the
suffix $y$ in~\eqref{eqn:infixlang} may be bounded uniformly in $L$,
independent of $w$. We need a slight strengthening of this in that the
prefix $x$ may actually be assumed to depend only on $L$, not on
$v$. We prove this by giving an equivalent condition on DFAs accepting
the language $L$:
\begin{lemma}
  \label{lem:infixautomaton}
  Let $\cA = (\Sigma,Q,q_0,\delta,A)$ be a deterministic finite automaton in
  which every state is reachable. Then $L(\mathcal{A})$ is infix complete if, and
  only if, some closed reachability class contains an accepting
  state. In this case there is a word $x \in \Sigma^*$ and a $k \geq
  0$ such that for every $v \in \Sigma^*$ there is a $y \in
  \Sigma^{\leq k}$ with $xvy \in L$.
\end{lemma}
\begin{proof}
  Let us first assume that some closed reachability class $C$ contains
  an accepting state $q \in C \cap A$. Since all states are assumed to
  be reachable, there is a string $x \in \Sigma^*$ with
  $\hat\delta(q_0,x) = q$. Now for every $v \in \Sigma^*$, $q_v
  \coloneqq \hat\delta(q,v) \in C$, and because $q$ is reachable from
  every state in $C$ there is a $y_v \in \Sigma^*$ with
  $\hat\delta(q_v,y_v) = q$ and therefore
  \[
  \hat\delta(q_0,xvy_v) = \hat\delta(q,vy_v) = \hat\delta(q_v,y_v) = q \in A,
  \]
  so $xvy \in L(\cA)$. Since $C$ is finite and the string $y$ depends
  only on $q_v \in C$ we may assume $\abs{y_v} \leq k$ for some $k$
  depending only on $\cA$.

  For the converse direction, we adversarially construct a string $v$
  such that for every $q \in Q$, $\hat\delta(q,v)$ is an element of
  some closed reachability class. Then for all $x, y \in \Sigma^*$,
  $\hat\delta(q,xvy)$ is an element of a closed reachability class,
  and if no such class contains an accepting state we get $xvy \not\in
  L(\cA)$, so $L(\cA)$ is not infix complete.

  To construct $v$ we first pick, for every state $q \in Q$, a string
  $v_q \in \Sigma^*$ such that $\hat\delta(q,v_q)$ is an element of
  some closed reachability class (such a $v_q$ must exist because
  $Q$ is finite). We enumerate the states as $Q = \{q_0,\ldots,q_r\}$
  and set
  \[
  v_0 \coloneqq \emptyword
  \qtext{and}
  v_{i+1} \coloneqq v_iv_{\hat\delta(q_i,v_i)}\text{ for }i=0,\ldots,r.
  \]
  Then $v_{r+1}$ has the desired property: If the automaton starts in
  some state $q_i$, then $v_i$ takes it to some state
  $q = \hat\delta(q_i,v_i)$, and from there
  $v_{\hat\delta(q_i,v_i)} = v_q$ takes it to some closed reachability
  class $C$. Since this can not be left, also
  $\hat\delta(q_i,v_{r+1}) \in C$.
\qed
\end{proof}
\begin{corollary}
  Given a DFA $\cA$, it is decidable whether $L(\cA)$ is infix complete.
\end{corollary}
In fact, for a regular language $L$ the language $\infix(L)$ is again
regular, and one can easily compute a DFA for $\infix(L)$ given a DFA
for $L$.

An $\omega$-word $\xi \in \Sigma^\omega$ is called \emph{rich} if for every $v \in \Sigma^*$ 
there are $x\in \Sigma^*$ and $y \in \Sigma^\omega$ such that $\xi = xvy$ \cite{Staiger1998}.
Our main result on the density of regular $\omega$-languages now reads:
\begin{theorem}
  \label{thm:infwords}
  Let $L \subseteq \Sigma^\omega$ be a regular $\omega$-language. Then
  $\Prob(L) > 0$ if and only if $L$ contains a rich $\omega$-word. 
  Moreover, given a suitable description of $L$ it is decidable whether $\Prob(L) > 0$ or not.
\end{theorem}
\begin{proof}
  Let $L = \bigcup_i U_iV_i^\omega$ with regular $U_i, V_i \subseteq \Sigma^*$.
  Obviously $L$ contains a rich $\omega$-word if, and only
  if, one of the languages $U_iV_i^\omega$ does, which is the case if,
  and only if, $V_i^\omega$ contains such a word. On
  the other hand,
  \(
  \Prob(L) \leq \sum_{i=1}^k \Prob(U_iV_i^\omega),
  \)
  so $\Prob(L) > 0$ if and only if $\Prob(U_iV_i^\omega) > 0$ for some
  $i$. Therefore without loss of generality we may assume that $L =
  UV^\omega$ for regular $U,V \subseteq \Sigma^*$.

  $V^\omega$ contains a rich $\omega$-word if and only if $V^*$ is
  infix complete: Since $\Sigma^*$ is countable, if $V^*$ is infix
  complete we may take a word $w_v = xvy \in V$ for each $v \in
  \Sigma^*$ and concatenate these words to get a rich
  $\omega$-word in $V^\omega$. On the other hand, if $\xi \in
  V^\omega$ is rich, then every $w \in \Sigma^*$ is contained in some
  finite prefix of $\xi$, which in turn is a prefix of some word in
  $V^*$.

  We now show the following: If $U,V \subseteq \Sigma^*$ are regular
  languages and $V^*$ is infix complete there is a word $x \in
  \Sigma^*$ such that
  \[
    \Prob\big( \zeta \in L \sthat x \preceq \zeta ) = 1
    \qquad\text{or, equivalently,}
    \qquad
    \Prob\big(\{\zeta \in \Sigma^\omega \sthat x \preceq \zeta\}
    \setminus L\big) = 0
  \]
  Since $(V^*)^\omega = V^\omega$ we may actually assume that $V$
  itself is infix complete. Pick a DFA $\cA = (Q,q_0,\delta,A)$ with
  $L(\cA) = V$. By Lemma~\ref{lem:infixautomaton} there is a closed
  reachability class $C \subseteq Q$ that contains an accepting state
  $\tilde q \in C \cap A$. Let $w = a_1\cdots a_\ell \in \Sigma^*$ be
  any word such that $\hat\delta(q_0,w) \in C$. We define a new
  automaton $\cA' = (\Sigma, Q',q_0,\delta',A)$ by $Q' \coloneqq Q \cup \{
  q'_1,\ldots,q'_{\ell} \}$ and
  \[
  \begin{split}
    \delta'(q,a) &= \delta(q,a)\text{ if }q \not\in \{ \tilde q,
    q'_1,\ldots,q'_{\ell} \},
    \\
    \delta'(\tilde q,a) &= \begin{cases}
      q_1 &\text{if }a=a_1
      \\
      \delta(\tilde q,a)&\text{if }a \not= a_1
    \end{cases}
    \\
    \delta'(q'_i,a) &= \begin{cases}
      q'_{i+1}&\text{if }a=a_{i+1}\text{ and }i < \ell
      \\
      \hat\delta(\tilde q,a_1\cdots a_i a)&\text{if }a\not= a_{i+1}
    \end{cases}
    \\
    \delta'(q'_\ell,a) &= \delta(\hat\delta(q_0,w),a)
  \end{split}
  \]
  Then $C \cup \{ q'_1,\ldots,q'_{\ell} \}$ is a closed reachability
  class in $\cA'$ and if $\zeta = z_0z_1\cdots \in \Sigma^\omega$ is
  an $\omega$-word such that $w \preceq \zeta$ and
  $\hat\delta'(q_0,z_0\cdots z_k) = q'_\ell$ for infinitely many $k
  \geq 0$ (i.e.~when reading the $\omega$-word $\zeta$, the automaton
  $\cA'$ passes through the state $q'_\ell$ infinitely
  often), then $\zeta \in V^\omega$. In fact, let $k_1 < k_2 < \cdots$
  be chosen such that $\hat\delta'(q_0,z_0\cdots z_{k_i}) = q'_\ell$
  for every $i \geq 1$. Then each of the subwords
  \[
    z_0\cdots z_{k_1-\ell},\quad
    z_{k_1-\ell+1}\cdots z_{k_2-\ell},\quad
    z_{k_2-\ell+1}\cdots z_{k_3-\ell},\quad\cdots
  \]
  is in $V$, because it starts with the prefix $w$ and then takes the
  automaton $\cA$ to the (accepting) state $\tilde q$.

  Now, if $\zeta = \zeta_1 \zeta_2\zeta_3\cdots$ is a random
  $\omega$-word, the random variables $(X_t)_{t \geq 0}$ with $X_t =
  \hat{\delta'}(q_0,\zeta_1\cdots\zeta_t)$
  are a Markov chain with state space $Q'$, started in $q_0$ and with
  transition probabilities
  $p_{q,q'} = \abs{\Sigma}^{-1}$ if $\delta(q,a) = q'$ for some $a
         \in \Sigma$, and $p_{q,q'} =  0$ otherwise.
  In particular, the reachability relation $\to$ on the state space of
  this Markov chain is the same as the reachability relation
  $\leadsto$ of the automaton $\cA'$, and the closed classes of the
  Markov chain are exactly the closed reachability classes of
  $\cA'$. Thus $q'_\ell$ is recurrent by Thm.~\ref{thm:recurrent},
  and with Cor.~\ref{cor:markovrecurr} we get
  \[
    \Prob(X_t = q'_{\ell} \text{ infinitely often} \sthat w \preceq \zeta)
    =
    \Prob(X_t \in C \text{ for some } t \geq 0 \sthat w \preceq
    \zeta) = 1,
  \]
  because if $w \preceq \zeta$ then $X_\ell = \hat\delta(q_0,w) = \tilde q \in
  C$. Finally, if $u \in U$ is any word in $U$ then
  \(
    \Prob\big( \zeta \in L \sthat uw \preceq \zeta ) = 1.
  \)
\end{proof}

    \section{Conclusion}
We gave decidable characterisations for sparseness of regular tree
languages and of regular $\omega$-languages, both in terms of excluded
subtrees and in terms of automata accepting these languages.

By~\cite{Niwinski2023}, sparseness is decidable also for regular
languages of \emph{infinite} trees, allowing for probabilistic model
checking not just for linear temporal logic, but also for computation
tree logic (CTL), opening the route to many further applications in
model checking (cf.~\cite[Ch.~6]{Baier2008}). However, Niwinski et
al.~solve the more general problem of exactly computing the measure of
a regular infinite tree language, at prohibitively high computational
cost. Focussing on sparseness is likely to allow for more efficient
algorithms, as exemplified by the linear time algorithm from \Cref{cor:trees:decidable}.

Several known graph properties, such as being series-parallel or, more
generally, having bounded tree-width, imply that graphs with these
properties can be encoded in labelled trees in such a way that the
original graph can be MSO interpreted in the tree
(cf.~\cite[Ch.~11.4]{Flum2006}). Our results on sparseness of regular
tree languages can therefore be translated into sparseness conditions
for MSO-definable properties of graphs in these graph
classes. However, since a given graph may in general be interpretable
in many different trees, this yields sparseness with respect to some
non-uniform notion of density, prompting for a closer investigation.

Another interesting direction for future work would be more general
probabilistic models. While our methods (as well as those
in~\cite{Courcoubetis1995} and~\cite{Niwinski2023}) easily generalise
from independent letters to Markov chains, it is not immediately clear
if this can be done also for strings or trees generated by hidden
Markov models (HMMs, cf.~\cite{Yoon2009}) because in this case,
reachability classes of the automaton are no longer the same as
communicating classes of the resulting Markov chain.

    \bibliographystyle{alphaurl}
    \bibliography{bibliography}

\begin{thebibliography}{GKLT69}

\bibitem[AKNS18]{Atserias2018}
Albert Atserias, Stephan Kreutzer, and Marcos Noy~Serrano.
\newblock On zero-one and convergence laws for graphs embeddable on a fixed
  surface.
\newblock In {\em 45th International Colloquium on Automata, Languages, and
  Programming (ICALP 2018): July 9-13, 2018, Prague, Czech Republic}, pages
  1--14. Schloss Dagstuhl-Leibniz-Zentrum f{\"u}r Informatik, 2018.
\newblock \href {https://doi.org/10.4230/LIPIcs.ICALP.2018.116}
  {\path{doi:10.4230/LIPIcs.ICALP.2018.116}}.

\bibitem[AKST19]{Asada2019}
Kazuyuki Asada, Naoki Kobayashi, Ryoma Sin'ya, and Takeshi Tsukada.
\newblock Almost every simply typed lambda-term has a long beta-reduction
  sequence.
\newblock {\em Log. Methods Comput. Sci.}, 15(1), 2019.
\newblock \href {https://doi.org/10.23638/LMCS-15(1:16)2019}
  {\path{doi:10.23638/LMCS-15(1:16)2019}}.

\bibitem[BJ09]{Yoon2009}
Yoon Byung-Jun.
\newblock Hidden {Markov} models and their applications in biological sequence
  analysis.
\newblock {\em Current Genomics}, 10(6):402--415, 2009.
\newblock \href {https://doi.org/10.2174/138920209789177575}
  {\path{doi:10.2174/138920209789177575}}.

\bibitem[BKL08]{Baier2008}
C.~Baier, J.P. Katoen, and K.G. Larsen.
\newblock {\em Principles of Model Checking}.
\newblock MIT Press, 2008.

\bibitem[BS09]{Benedikt09}
Michael Benedikt and Luc Segoufin.
\newblock Towards a characterization of order-invariant queries over tame
  graphs.
\newblock {\em J. Symb. Log.}, 74(1):168--186, 2009.
\newblock \href {https://doi.org/10.2178/jsl/1231082307}
  {\path{doi:10.2178/jsl/1231082307}}.

\bibitem[Bü60]{Buechi1960}
J.~Richard Büchi.
\newblock Weak second-order arithmetic and finite automata.
\newblock {\em Mathematical Logic Quarterly}, 6(1-6):66--92, 1960.
\newblock \href {https://doi.org/https://doi.org/10.1002/malq.19600060105}
  {\path{doi:https://doi.org/10.1002/malq.19600060105}}.

\bibitem[Bü90]{Buechi1990}
J.~Richard Büchi.
\newblock {\em On a Decision Method in Restricted Second Order Arithmetic},
  pages 425--435.
\newblock Springer New York, New York, NY, 1990.
\newblock \href {https://doi.org/10.1007/978-1-4613-8928-6_23}
  {\path{doi:10.1007/978-1-4613-8928-6_23}}.

\bibitem[CY95]{Courcoubetis1995}
Costas Courcoubetis and Mihalis Yannakakis.
\newblock The complexity of probabilistic verification.
\newblock {\em Journal of the ACM (JACM)}, 42(4):857--907, 1995.
\newblock \href {https://doi.org/10.1145/210332.210339}
  {\path{doi:10.1145/210332.210339}}.

\bibitem[Dev86]{Devroye1986}
Luc Devroye.
\newblock A note on the height of binary search trees.
\newblock {\em Journal of the ACM (JACM)}, 33(3):489--498, 1986.
\newblock \href {https://doi.org/10.1145/5925.5930}
  {\path{doi:10.1145/5925.5930}}.

\bibitem[DK19]{Dawar2019}
Anuj Dawar and Eryk Kopczynski.
\newblock Logical properties of random graphs from small addable classes.
\newblock {\em Log. Methods Comput. Sci.}, 15(3), 2019.
\newblock \href {https://doi.org/10.23638/LMCS-15(3:4)2019}
  {\path{doi:10.23638/LMCS-15(3:4)2019}}.

\bibitem[Don70]{Doner1970}
John Doner.
\newblock Tree acceptors and some of their applications.
\newblock {\em J. Comput. Syst. Sci.}, 4(5):406--451, 1970.
\newblock \href {https://doi.org/10.1016/S0022-0000(70)80041-1}
  {\path{doi:10.1016/S0022-0000(70)80041-1}}.

\bibitem[EF99]{ebbflu99}
Heinz-Dieter Ebbinghaus and Jörg Flum.
\newblock {\em Finite Model Theory}.
\newblock Perspectives in Mathematical Logic. Springer, 2nd edition, 1999.
\newblock \href {https://doi.org/10.1007/3-540-28788-4}
  {\path{doi:10.1007/3-540-28788-4}}.

\bibitem[Elg61]{Elgot1961}
Calvin~C. Elgot.
\newblock Decision problems of finite automata design and related arithmetics.
\newblock {\em Transactions of the American Mathematical Society},
  98(1):21--51, 1961.
\newblock \href {https://doi.org/10.2307/1993511} {\path{doi:10.2307/1993511}}.

\bibitem[Fag76]{Fagin1976}
Ronald Fagin.
\newblock Probabilities on finite models.
\newblock {\em The Journal of Symbolic Logic}, 41(1):50--58, 1976.
\newblock \href {https://doi.org/10.2307/2272945} {\path{doi:10.2307/2272945}}.

\bibitem[FG06]{Flum2006}
J{\"{o}}rg Flum and Martin Grohe.
\newblock {\em Parameterized Complexity Theory}.
\newblock Texts in Theoretical Computer Science. An {EATCS} Series. Springer,
  2006.
\newblock \href {https://doi.org/10.1007/3-540-29953-X}
  {\path{doi:10.1007/3-540-29953-X}}.

\bibitem[FS09]{Flajolet2009}
Philippe Flajolet and Robert Sedgewick.
\newblock {\em Analytic Combinatorics}.
\newblock Cambridge University Press, 2009.
\newblock \href {https://doi.org/https://doi.org/10.1017/CBO9780511801655}
  {\path{doi:https://doi.org/10.1017/CBO9780511801655}}.

\bibitem[GKLT69]{Glebskij1969}
Y.V. Glebskij, D.I. Kogan, M.I. Liogon'kij, and V.A. Talanov.
\newblock Range and degree of realizability of formulas in the restricted
  predicate calculus.
\newblock {\em Cybernetics}, 5:142--154, 1969.
\newblock \href {https://doi.org/10.1007/BF01071084}
  {\path{doi:10.1007/BF01071084}}.

\bibitem[Knu97]{taocp1}
Donald~E. Knuth.
\newblock {\em The Art of Computer Programming: Fundamental Algorithms},
  volume~I.
\newblock Addison-Wesley, 3rd edition, 1997.

\bibitem[Kog19]{Koga2019}
Toshihiro Koga.
\newblock On the density of regular languages.
\newblock {\em Fundam. Informaticae}, 168(1):45--49, 2019.
\newblock \href {https://doi.org/10.3233/FI-2019-1823}
  {\path{doi:10.3233/FI-2019-1823}}.

\bibitem[McC02]{McColm2002}
Gregory~L. McColm.
\newblock Mso zero-one laws on random labelled acyclic graphs.
\newblock {\em Discrete Mathematics}, 254(1):331--347, 2002.
\newblock \href {https://doi.org/https://doi.org/10.1016/S0012-365X(01)00375-2}
  {\path{doi:https://doi.org/10.1016/S0012-365X(01)00375-2}}.

\bibitem[MZ21]{Malyshkin2021}
Y.A. Malyshkin and M.E. Zhukovskii.
\newblock Mso 0-1 law for recursive random trees.
\newblock {\em Statistics \& Probability Letters}, 173:109061, 2021.
\newblock \href {https://doi.org/https://doi.org/10.1016/j.spl.2021.109061}
  {\path{doi:https://doi.org/10.1016/j.spl.2021.109061}}.

\bibitem[Nor97]{Norris1997}
James~R. Norris.
\newblock {\em Markov Chains}.
\newblock Cambridge University Press, 1997.
\newblock \href {https://doi.org/10.1017/CBO9780511810633}
  {\path{doi:10.1017/CBO9780511810633}}.

\bibitem[NPS23]{Niwinski2023}
Damian Niwiński, Paweł Parys, and Michał Skrzypczak.
\newblock The probabilistic {Rabin} tree theorem*.
\newblock In {\em 2023 38th Annual ACM/IEEE Symposium on Logic in Computer
  Science (LICS)}, pages 1--13, 2023.
\newblock \href {https://doi.org/10.1109/LICS56636.2023.10175800}
  {\path{doi:10.1109/LICS56636.2023.10175800}}.

\bibitem[Sin15]{Sinya2015}
Ryoma Sin'ya.
\newblock An automata theoretic approach to the zero-one law for regular
  languages: Algorithmic and logical aspects.
\newblock In Javier Esparza and Enrico Tronci, editors, {\em Proceedings of the
  Sixth International Symposium on Games, Automata, Logics and Formal
  Verification, GandALF 2015, Genoa, Italy, 21-22nd September 2015}, volume 193
  of {\em {EPTCS}}, pages 172--185, 2015.
\newblock \href {https://doi.org/10.4204/EPTCS.193.13}
  {\path{doi:10.4204/EPTCS.193.13}}.

\bibitem[Spe13]{Spencer2013}
Joel Spencer.
\newblock {\em The strange logic of random graphs}, volume~22.
\newblock Springer Science \& Business Media, 2013.
\newblock \href {https://doi.org/10.1007/978-3-662-04538-1}
  {\path{doi:10.1007/978-3-662-04538-1}}.

\bibitem[Sta97]{Staiger1997}
Ludwig Staiger.
\newblock $\omega$-languages.
\newblock In {\em Handbook of Formal Languages: Volume 3 Beyond Words}, pages
  339--388. Springer Berlin Heidelberg, Berlin, Heidelberg, 1997.
\newblock \href {https://doi.org/10.1007/978-3-642-59126-6_6}
  {\path{doi:10.1007/978-3-642-59126-6_6}}.

\bibitem[Sta98]{Staiger1998}
Ludwig Staiger.
\newblock Rich $\omega$-words and monadic second-order arithmetic.
\newblock In Mogens Nielsen and Wolfgang Thomas, editors, {\em Computer Science
  Logic}, pages 478--490, Berlin, Heidelberg, 1998. Springer Berlin Heidelberg.
\newblock \href {https://doi.org/10.1007/BFb0028032}
  {\path{doi:10.1007/BFb0028032}}.

\bibitem[Tar72]{Tarjan1972}
Robert~Endre Tarjan.
\newblock Depth-first search and linear graph algorithms.
\newblock {\em {SIAM} J. Comput.}, 1(2):146--160, 1972.
\newblock \href {https://doi.org/10.1137/0201010} {\path{doi:10.1137/0201010}}.

\bibitem[Tho97]{Thomas1997}
Wolfgang Thomas.
\newblock Languages, automata, and logic.
\newblock In {\em Handbook of Formal Languages: Volume 3 Beyond Words}, pages
  389--455. Springer Berlin Heidelberg, Berlin, Heidelberg, 1997.
\newblock \href {https://doi.org/10.1007/978-3-642-59126-6_7}
  {\path{doi:10.1007/978-3-642-59126-6_7}}.

\bibitem[Tra61]{Trakhtenbrot1961}
Boris~A. Trakhtenbrot.
\newblock Finite automata and the logic of single-place predicates.
\newblock {\em Dokl. Akad. Nauk SSSR}, 140(2):326--329, 1961.
\newblock URL: \url{https://www.mathnet.ru/eng/dan25511}.

\bibitem[TW68]{Thatcher1968}
James~W. Thatcher and Jesse~B. Wright.
\newblock Generalized finite automata theory with an application to a decision
  problem of second-order logic.
\newblock {\em Math. Syst. Theory}, 2(1):57--81, 1968.
\newblock \href {https://doi.org/10.1007/BF01691346}
  {\path{doi:10.1007/BF01691346}}.

\bibitem[Wil91]{Williams1991}
David Williams.
\newblock {\em Probability with Martingales}.
\newblock Cambridge University Press, 1991.
\newblock \href {https://doi.org/10.1017/CBO9780511813658}
  {\path{doi:10.1017/CBO9780511813658}}.

\end{thebibliography}

\end{document}